\documentclass[a4paper,onecolumn,11pt,accepted=2019-08-26]{quantumarticle}
\pdfoutput=1
\usepackage[utf8]{inputenc}
\usepackage[english]{babel}
\usepackage[T1]{fontenc}

\usepackage{amssymb,amsmath,amsthm,amsfonts,dsfont}
\usepackage[numbers,comma,sort&compress]{natbib}
\usepackage[colorlinks]{hyperref}
\usepackage[all]{hypcap}
\usepackage{url}
\usepackage{graphicx}
\usepackage[font=small]{caption}
\usepackage[labelformat=simple]{subcaption}
\usepackage{float}
\usepackage{algorithmic}
\usepackage[ruled,vlined,linesnumbered]{algorithm2e}
\usepackage{footnote}
\usepackage{enumitem}
\usepackage[margin=1in]{geometry}

\usepackage{xcolor}
\usepackage{pgfplots}

\usepackage{mathtools}
\mathtoolsset{centercolon}

\DeclarePairedDelimiterX\braket[2]{\langle}{\rangle}{#1 \delimsize\vert #2}

\newtheorem{theorem}{Theorem}

\newtheorem{lemma}{Lemma}

\newtheorem{corollary}{Corollary}

\newcommand{\eq}[1]{(\ref{eq:#1})}
\newcommand{\thm}[1]{\hyperref[thm:#1]{Theorem~\ref*{thm:#1}}}
\newcommand{\defn}[1]{\hyperref[defn:#1]{Definition~\ref*{defn:#1}}}
\newcommand{\lem}[1]{\hyperref[lem:#1]{Lemma~\ref*{lem:#1}}}
\newcommand{\prop}[1]{\hyperref[prop:#1]{Proposition~\ref*{prop:#1}}}
\newcommand{\fig}[1]{\hyperref[fig:#1]{Figure~\ref*{fig:#1}}}
\newcommand{\tab}[1]{\hyperref[tab:#1]{Table~\ref*{tab:#1}}}
\renewcommand{\sec}[1]{\hyperref[sec:#1]{Section~\ref*{sec:#1}}}
\newcommand{\append}[1]{\hyperref[append:#1]{Appendix~\ref*{append:#1}}}
\newcommand{\cor}[1]{\hyperref[cor:#1]{Corollary~\ref*{cor:#1}}}

\newcommand{\N}{\mathbb{N}}

\newcommand{\R}{\mathbb{R}}
\newcommand{\C}{\mathbb{C}}

\DeclareMathOperator{\Sym}{Sym}

\newcommand{\rev}{\mathrm{rev}}
\newcommand{\rand}{\mathrm{rand}}
\newcommand{\comm}{\mathrm{comm}}
\newcommand{\demp}{\mathrm{demp}}
\newcommand{\remp}{\mathrm{remp}}

\newcommand{\comment}[1]{}

\newcommand{\norm}[1]{\left\lVert#1\right\rVert}

\pgfplotsset{
	log x ticks with fixed point/.style={
		xticklabel={
			\pgfkeys{/pgf/fpu=true}
			\pgfmathparse{exp(\tick)}%
			\pgfmathprintnumber[fixed relative, precision=3]{\pgfmathresult}
			\pgfkeys{/pgf/fpu=false}
		}
	},
	log y ticks with fixed point/.style={
		yticklabel={
			\pgfkeys{/pgf/fpu=true}
			\pgfmathparse{exp(\tick)}%
			\pgfmathprintnumber[fixed relative, precision=3]{\pgfmathresult}
			\pgfkeys{/pgf/fpu=false}
		}
	}
}

\pgfdeclareplotmark{hexagon}
{%
	\pgfpathmoveto{\pgfqpointpolar{0}{1.1547\pgfplotmarksize}}
	\pgfpathlineto{\pgfqpointpolar{60}{1.1547\pgfplotmarksize}}
	\pgfpathlineto{\pgfqpointpolar{120}{1.1547\pgfplotmarksize}}
	\pgfpathlineto{\pgfqpointpolar{180}{1.1547\pgfplotmarksize}}
	\pgfpathlineto{\pgfqpointpolar{240}{1.1547\pgfplotmarksize}}
	\pgfpathlineto{\pgfqpointpolar{300}{1.1547\pgfplotmarksize}}
	\pgfpathclose
	\pgfusepathqstroke
}

\pgfdeclareplotmark{hexagon*}
{%
	\pgfpathmoveto{\pgfqpointpolar{0}{1.1547\pgfplotmarksize}}
	\pgfpathlineto{\pgfqpointpolar{60}{1.1547\pgfplotmarksize}}
	\pgfpathlineto{\pgfqpointpolar{120}{1.1547\pgfplotmarksize}}
	\pgfpathlineto{\pgfqpointpolar{180}{1.1547\pgfplotmarksize}}
	\pgfpathlineto{\pgfqpointpolar{240}{1.1547\pgfplotmarksize}}
	\pgfpathlineto{\pgfqpointpolar{300}{1.1547\pgfplotmarksize}}
	\pgfpathclose
	\pgfusepathqfillstroke
}

\usepackage{etoolbox}

\makeatletter
\patchcmd{\@maketitle}{\raggedright}{\centering}{}{}
\patchcmd{\@maketitle}{\raggedright}{\centering}{}{}
\makeatother

\begin{document}

\title{Faster quantum simulation by randomization}
\author[aff1,aff2,aff3]{Andrew M.\ Childs}
\author[aff2,aff3,aff4]{Aaron Ostrander}
\author[aff1,aff2,aff3]{Yuan Su}
\affiliation[aff1]{Department of Computer Science, University of Maryland}
\affiliation[aff2]{Institute for Advanced Computer Studies, University of Maryland}
\affiliation[aff3]{Joint Center for Quantum Information and Computer Science, University of Maryland}
\affiliation[aff4]{Department of Physics, University of Maryland}

\maketitle

\begin{abstract}
Product formulas can be used to simulate Hamiltonian dynamics on a quantum computer by approximating the exponential of a sum of operators by a product of exponentials of the individual summands. This approach is both straightforward and surprisingly efficient. We show that by simply randomizing how the summands are ordered, one can prove stronger bounds on the quality of approximation for product formulas of any given order, and thereby give more efficient simulations. Indeed, we show that these bounds can be asymptotically better than previous bounds that exploit commutation between the summands, despite using much less information about the structure of the Hamiltonian. Numerical evidence suggests that the randomized approach has better empirical performance as well.
\end{abstract}

\section{Introduction}
\label{sec:intro}

Simulating quantum dynamics is one of the major potential applications of quantum computers. The apparent intractability of simulating quantum dynamics with a classical computer led Feynman \cite{Fey82} and others to propose the idea of quantum computation. Lloyd gave the first explicit quantum algorithm for simulating the dynamics of local Hamiltonians \cite{Llo96}, and later work showed that the more general class of sparse Hamiltonians can also be simulated efficiently \cite{AT03}. Quantum simulation can be applied to understand the behavior of various physical systems---including many-body physics \cite{RWS12}, quantum chemistry \cite{WBCHT14,Pou15,BMWAW15}, and quantum field theory \cite{JLP12}---and designing new quantum algorithms \cite{CCDFGS03,FGG07,HHL09,BCOW17,BS17}.

The main ingredient in Lloyd's algorithm is the Lie product formula, which provides a first-order approximation to the exponential of a sum as a product of exponentials of the summands. Given Hermitian operators $H_1,\ldots,H_L$ (which we refer to as the summands of the Hamiltonian $H = \sum_{j=1}^L H_j$) and a complex number $\lambda$, the Lie product formula
\begin{equation}
S_1(\lambda):=\prod_{j=1}^{L}\exp(\lambda H_j),
\end{equation}
approximates the exponentiation
\begin{equation}
\label{eq:idealexp}
V(\lambda):=\exp\biggl(\lambda\sum_{j=1}^{L}H_j\biggr)
\end{equation}
in the sense that $V(\lambda) \approx S_1({\lambda}/{r})^r$ for large $r$. Suzuki systematically extended this formula to give a $(2k)$th-order approximation $S_{2k}$, defined recursively by
\begin{equation}
\begin{aligned}
\label{eq:pf2k}
S_{2}(\lambda)&:=\prod_{j=1}^{L}\exp\biggl(\frac{\lambda}{2}H_{j}\biggr)\prod_{j=L}^{1}\exp\biggl(\frac{\lambda}{2}H_{j}\biggr)\\
S_{2k}(\lambda)&:=S_{2k-2}(p_{k}\lambda)^2 \, S_{2k-2}((1-4p_{k})\lambda) \, S_{2k-2}(p_{k}\lambda)^2
\end{aligned}
\end{equation}
with $p_{k}:=1/(4-4^{1/(2k-1)})$ \cite{Suz91}. Again we have $V(\lambda) \approx S_{2k}({\lambda}/{r})^r$ for large $r$, and the approximation obtained with a given value of $r$ improves as $k$ increases (albeit with a prefactor that grows exponentially in $k$). We refer to all such formulas as \emph{product formulas}. When they are used for quantum simulation, $H$ is chosen to be the Hamiltonian and $\lambda=-it$, where $t$ is the evolution time. Although other approaches to quantum simulation have better proven asymptotic performance as a function of various parameters \cite{BCK15,BC12,BCCKS13,BCCKS14,LC17,LC16,HHKL18}, product formulas perform well in practice \cite{CMNRS18} and are widely used in experimental implementations \cite{Bar15,BCC06,Lan11} due to their simplicity and the fact that they do not require any ancilla qubits.

The main challenge in applying product formulas to quantum simulation is to choose the number of segments $r$ to ensure the simulation error is at most some allowed threshold $\epsilon$. To simulate $H=\sum_{j=1}^{L}H_j$ for time $t$, rigorous error analysis shows that
\begin{equation}
r_{1,\text{det}}=O\biggl(\frac{(t\Lambda L)^2}{\epsilon}\biggr)
\end{equation}
suffices to ensure error at most $\epsilon$ for the first-order formula and
\begin{equation}
r_{2k,\text{det}}=O\biggl(\frac{(t\Lambda L)^{1+\frac{1}{2k}}}{\epsilon^{\frac{1}{2k}}}\biggr)
\end{equation}
suffices for $(2k)$th order \cite{BACS05}, where $\Lambda := \max_{j}\norm{H_{j}}$ is a spectral-norm upper bound on the summands and $\text{det}$ indicates that these formulas are constructed deterministically. However, numerical simulations suggest that the product formula algorithm can perform significantly better in practice than the best proven error bounds demonstrate \cite{BMWAW15,RWS12,RWSWT17,CMNRS18}. Indeed, recent work suggests that it can even asymptotically outperform more sophisticated simulation algorithms with better proven running times \cite{CMNRS18}. This dramatic gap between the provable and the actual behavior of product formula simulation suggests that it may be possible to significantly improve their analysis, and thereby give more efficient algorithms for quantum simulation.

It is sometimes possible to improve the analysis of product formulas using further information about the form of the Hamiltonian. In particular, the cost of simulation can be reduced when many pairs of summands commute \cite{Llo96,BMWAW15,CMNRS18}. However, this approach can only be applied for structured Hamiltonians that contain many commuting summands. Furthermore, the best known bounds of this type give only modest improvement, remaining orders of magnitude away from the empirical performance even in cases where many summands commute \cite{CMNRS18}.

Randomization can be a powerful tool for improving the performance of quantum simulation algorithms. For example, Poulin et al.\ gave improved simulations of time-dependent Hamiltonians by sampling the Hamiltonian at random times \cite{PQSV11}. Closer in spirit to the present paper, Zhang studied the effect of randomizing the ordering and/or duration of evolutions in a product formula, showing in particular that randomly ordering the summands in the first-order formula in either forward or reverse order can give an improved algorithm \cite{Zha10}.

In this paper, we explore a closely related approach for higher-order product formulas, which can achieve significantly better asymptotic performance.
Specifically, we analyze the effect of randomly permuting the summands. The resulting algorithm is not much more complicated than a deterministic product formula, but the savings in the simulation cost are substantial. For any permutation $\sigma\in\Sym(L)$ of the $L$ summands, let
\begin{equation}
\label{eq:permpf}
\begin{aligned}
S_2^\sigma(\lambda)&:=\prod_{j=1}^{L}\exp\biggl(\frac{\lambda}{2}H_{\sigma(j)}\biggr)\prod_{j=L}^{1}\exp\biggl(\frac{\lambda}{2}H_{\sigma(j)}\biggr)\\
S_{2k}^\sigma(\lambda)&:=[S_{2k-2}^\sigma(p_{k}\lambda)]^{2}S_{2k-2}^\sigma((1-4p_{k})\lambda)[S_{2k-2}^\sigma(p_{k}\lambda)]^{2}.
\end{aligned}
\end{equation}
We show that the $(2k)$th-order randomized simulation has error
\begin{equation}
\norm{\mathcal{V}(-it)-\biggl(\frac{1}{L!}\sum_{\sigma\in\Sym(L)}\mathcal{S}_{2k}^\sigma\bigl(-it/r\bigr)\biggr)^r}_\diamond
= O\biggl(\frac{(\Lambda tL)^{4k+2}}{ r^{4k+1}}+\frac{(\Lambda t)^{2k+1}L^{2k}}{ r^{2k}}\biggr).
\end{equation}
where $\mathcal{V}(-it)$ and $\mathcal{S}_{2k}^\sigma(-it/r)$ are quantum channels describing the unitary transformation $V(-it)$ and the random unitary $S_{2k}^\sigma(-it/r)$, respectively, and $\norm{\cdot}_\diamond$ is the diamond norm (defined in \sec{strategy}).

Our analysis uses a mixing lemma of Campbell and Hastings \cite{Has17,Cam17} to bound the diamond norm distance from the ideal evolution. (Even for the first-order case, this improves over the analysis of Zhang, which uses similar methods but only bounds the trace distance from the ideal final state \cite{Zha10}, a metric that does not account for entanglement with a reference system.) Informally, the lemma of \cite{Has17,Cam17} states that if we can approximate a desired operation as the average over some set of operations, then the overall error depends linearly on the error in the average operation but only quadratically on the error in any individual operation. Standard error bounds for product formulas do not depend on how the summands are ordered, but we show that randomizing the ordering gives a more accurate average evolution.
We motivate this approach in \sec{strategy}, where we consider the effect of randomizing how the summands are ordered in the simple case of the first-order formula. Assuming $\Lambda:=\max_j \norm{H_j}$ is constant, the randomized first-order algorithm has gate complexity $g_{1}^\text{rand}=O\big(t^{1.5}L^{2.5}/\epsilon^{0.5}\big)$, improving over $g_{1}^{\det}=O\bigl(t^2L^3/\epsilon\bigr)$ in the deterministic case.

Analyzing the effect of randomization on higher-order formulas is more challenging. For terms of order at most $L$ in the Taylor expansion of a product formula, the majority of the error comes from terms in which no summands are repeated. We call such contributions \emph{nondegenerate terms}. In \sec{rand}, we give a combinatorial argument to compute nondegenerate terms of the average evolution $\frac{1}{L!}\sum_{\sigma\in\Sym(L)}S_{2k}^\sigma(\lambda)$ in closed form. (In fact, we prove a more general result that applies to the average evolution as a special case.) As a corollary, we show that the nondegenerate terms completely cancel in the randomized product formula.

\sec{errorbound} presents our main technical result, an upper bound on the error in a randomized higher-order product formula simulation. This bound follows by using the mixing lemma to combine an error bound for the average evolution operator with standard product formula error bounds for the error of the individual terms. \sec{performance} discusses the overall performance of the resulting algorithm and compares it with deterministic approaches. For the $(2k)$th-order product formula, assuming $\Lambda:=\max_j \norm{H_j}$ is constant, our randomized Hamiltonian simulation algorithm has complexity
\begin{equation}
g_{2k}^\text{rand}
=\max\bigg\{O\biggl(t L^2\biggl(\frac{t L}{\epsilon}\biggr)^{\frac{1}{4k+1}}\biggr),
O\biggl(t L^2\biggl(\frac{t}{\epsilon}\biggr)^{\frac{1}{2k}}\biggr)\bigg\},
\end{equation}
compared to $g_{2k}^{\det}=O\bigl(t L^2(tL/\epsilon)^{\frac{1}{2k}}\bigr)$ in the deterministic case. Thus our algorithm always improves the dependence on $L$ and sometimes achieves better dependence on $t$ and $\epsilon$ as well.

We also show in \sec{performance} that our bound can outperform a previous bound that takes advantage of the structure of the Hamiltonian. Specifically, we compare our randomized product formula algorithm with the deterministic algorithm using the commutator bound of \cite{CMNRS18} for a one-dimensional Heisenberg model in a random magnetic field. We find that over a significant range of parameters, the randomized algorithm has better proven performance, despite using less information about the form of the Hamiltonian.

In light of the large gap between proven and empirical performance of product formulas, it is natural to ask whether randomized product formulas still offer an improvement under the best possible error bounds. To address this question, we present numerical comparisons of the deterministic and randomized product formulas in \sec{numerics}. In particular, we show that the randomized approach can sometimes outperform the deterministic approach even with respect to their empirical performance.

Finally, we conclude in \sec{discussion} with a brief discussion of the results and some open questions.

\section{The power of randomization}
\label{sec:strategy}

To see how randomness can improve a product formula simulation, consider a simple Hamiltonian expressed as a sum of two operators, $H=H_1+H_2$. The Taylor expansion of the first-order formula as a function of $\lambda \in \C$ is
\begin{equation}
S_1(\lambda)=\exp(\lambda H_1)\exp(\lambda H_2)
=I+\lambda(H_1+H_2)+\frac{\lambda^2}{2}(H_1^2+2H_1H_2+H_2^2)+O(\lambda^3),
\end{equation}
whereas the Taylor series of the ideal evolution is
\begin{equation}
V(\lambda)=\exp((H_1+H_2)\lambda)
=I+\lambda(H_1+H_2)+\frac{\lambda^2}{2}(H_1^2+H_1H_2+H_2H_1+H_2^2)+O(\lambda^3).
\end{equation}
Using the triangle inequality, we can bound the spectral-norm error as
\begin{equation}
\norm{V(\lambda)-S_1(\lambda)}\leq \norm{[H_1,H_2]}\frac{|\lambda|^2}{2}+O((\Lambda|\lambda|)^3),
\end{equation}
where $\Lambda:=\max\{\norm{H_1},\norm{H_2}\}$. Since $H_1$ and $H_2$ need not commute, $S_1(\lambda)$ approximates $V(\lambda)$ to first order in $\lambda$, as expected.

It is clearly impossible to approximate $V(\lambda)$ to second order using a product of only two exponentials of $H_1$ and $H_2$: any such product can have only one of the products $H_1H_2$ and $H_2H_1$ in its Taylor expansion, whereas $V(\lambda)$ contains both of these products in its second-order term. However, we can obtain both products by taking a uniform mixture of $S_1(\lambda)$ and
\begin{equation}
S_1^\rev(\lambda):=\exp(\lambda H_2)\exp(\lambda H_1).
\end{equation}
Indeed, a simple calculation shows that
\begin{equation}
\norm{V(\lambda)-\frac{1}{2}\bigl(S_1(\lambda)+S_1^\rev(\lambda)\bigr)} = O\bigl((\Lambda|\lambda|)^3\bigr).
\end{equation}
However, $\bigl(S_1(-it)+S_1^\rev(-it)\bigr)/2$ is not a unitary operation in general. We could in principle implement a linear combination of unitaries using the techniques of \cite{BCCKS13}, but such an approach would use ancillas and could have high cost, especially when the Hamiltonian contains many summands. A simpler approach is to apply one of the two operations $S_1(-it)$ and $S_1^\rev(-it)$ chosen uniformly at random (as in Algorithm~2 of \cite{Zha10}), thereby implementing a quantum channel that gives a good approximation to the desired evolution.

We now introduce some notation that is useful to analyze the performance of randomized product formulas. Let $X$ be a matrix acting on a finite-dimensional Hilbert space $\mathcal{H}$. We write $\norm{X}$ for its spectral norm (the largest singular value) and $\norm{X}_1$ for its trace norm (the sum of its singular values, i.e., its Schatten 1-norm). Let $\mathcal{E}\colon X\mapsto \mathcal{E}(X)$ be a linear map on the space of matrices on $\mathcal{H}$. The diamond norm of $\mathcal{E}$ is
\begin{equation}
\norm{\mathcal{E}}_\diamond:=\max\{\norm{(\mathcal{E}\otimes\mathds{1}_\mathcal{H})(Y)}_1:\norm{Y}_1\leq 1\},
\end{equation}
where the maximization is taken over all matrices $Y$ on $\mathcal{H}\otimes\mathcal{H}$ satisfying $\norm{Y}_1\leq 1$.

The following mixing lemma bounds how well we can approximate a unitary operation using a random unitary channel. Specifically, the error is linear in the distance between the target unitary and the average of the random unitaries, and only quadratic in the distance between the target unitary and each individual random unitary.

\begin{lemma}[Mixing lemma \cite{Cam17,Has17}]
	\label{lem:mix}
	Let $V$ and $\{U_j\}$ be unitary matrices, with associated quantum channels $\mathcal{V}\colon \rho \mapsto V\rho V^\dagger$ and $\mathcal{U}_j\colon \rho \mapsto U_j\rho U_j^\dagger$, and let $\{p_j\}$ be a collection of positive numbers satisfying $\sum_jp_j=1$. Suppose that
	\begin{enumerate}[topsep=4pt,itemsep=0pt,label=\normalfont(\roman*),align=left,leftmargin=*]
		\item $\norm{U_j-V}\leq a$ for all $j$ and
		\item $\bigl\|(\sum_jp_jU_j)-V\bigr\| \leq b$.
	\end{enumerate}
	Then the average evolution $\mathcal{E}:=\sum_jp_j\mathcal{U}_j$ satisfies
	$
	\norm{\mathcal{E}-\mathcal{V}}_{\diamond}\leq a^2+2b
	$.
\end{lemma}

To simulate the Hamiltonian $H=H_1+H_2$ for time $t$, we divide the evolution into $r$ segments of duration $t/r$ and implement each segment via the random unitary operation
\begin{equation}
\frac{1}{2}\bigl(\mathcal{S}_1(-it/r)+\mathcal{S}_1^\rev(-it/r)\bigr)
\end{equation}
using one bit of randomness per segment, where $\mathcal{S}_1$ and $\mathcal{S}_1^\rev$ are the quantum channels associated with $S_1$ and $S_1^\rev$. Invoking the mixing lemma with $a=O\bigl({(\Lambda t)^2}/{r^2}\bigr)$ and $b=O\bigl({(\Lambda t)^3}/{r^3}\bigr)$, we find that
\begin{equation}
\norm{\mathcal{V}(-it/r)-\frac{1}{2}\bigl(\mathcal{S}_1(-it/r)+\mathcal{S}_1^\rev(-it/r)\bigr)}_\diamond
= O\bigg(\frac{(\Lambda t)^3}{r^3}\bigg).
\end{equation}
Since the diamond norm distance between quantum channels is subadditive under composition \cite[p.~178]{bk:wat}, the error of the entire simulation is
\begin{equation}
\norm{\mathcal{V}(-it)-\frac{1}{2^r}\bigl(\mathcal{S}_1(-it/r)+\mathcal{S}_1^\rev(-it/r)\bigr)^r}_\diamond
= O\bigg(\frac{(\Lambda t)^3}{r^2}\bigg).
\end{equation}
Thus the randomized first-order formula is effectively a second-order formula.

This approach easily extends to a sum of $L$ operators, again effectively making the first-order formula accurate to second order (cf.\ \cite{Zha10}, which shows the same result with respect to trace distance of the output state). Keeping track of all the prefactors, we find the following error bound for the randomized first-order formula.

\begin{theorem}[Randomized first-order error bound]
	\label{thm:frst}
	Let $\{H_j\}_{j=1}^{L}$ be Hermitian matrices. Let
	\begin{equation}
	V(-it):=\exp\biggl(-it\sum_{j=1}^{L}H_j\biggr)
	\end{equation}
	be the evolution induced by the Hamiltonian $H=\sum_{j=1}^{L}H_j$ for time $t\in\R$.
	Define
	\begin{equation}
	\label{eq:permpf1}
	\begin{aligned}
	S_1(\lambda)&:=\prod_{j=1}^{L}\exp(\lambda H_j) &&and &
	S_{1}^\rev(\lambda)&:=\prod_{j=L}^{1}\exp(\lambda H_j).
	\end{aligned}
	\end{equation}
	Let $r\in\N$ be a positive integer and $\Lambda:=\max\norm{H_j}$. Then
	\begin{equation}
	\label{eq:randpf1error}
	\norm{\mathcal{V}(-it)-\frac{1}{2^r}\bigl(\mathcal{S}_1(-it/r)+\mathcal{S}_1^\rev(-it/r)\bigr)^r}_\diamond
	\leq \frac{(\Lambda |t| L)^4}{r^3}
	\exp\biggl(2\frac{\Lambda|t|L}{r}\biggr)
	+ \frac{2(\Lambda|t|L)^3}{3r^2} \exp\biggl(\frac{\Lambda|t|L}{r}\biggr)
	\end{equation}
	where, for $\lambda=-it$, we associate channels $\mathcal{V}(\lambda)$, $\mathcal{S}_1(\lambda)$, and $\mathcal{S}_1^\rev(\lambda)$ with the unitaries $V(\lambda)$, $S_1(\lambda)$, and $S_1^\rev(\lambda)$, respectively.
\end{theorem}

To guarantee that the simulation error is at most $\epsilon$, we upper bound the right-hand side of \eq{randpf1error} by $\epsilon$ and solve for $r$. Assuming $\Lambda:=\max_j \norm{H_j}$ is constant, we find that it suffices to choose $r_{1}^\rand=O\bigl((tL)^{1.5}/\epsilon^{0.5}\bigr)$, giving a simulation algorithm with gate complexity $g_{1}^\rand=O\bigl(t^{1.5}L^{2.5}/\epsilon^{0.5}\bigr)$. In comparison, the gate complexity in the deterministic case is $g_{1}^{\det}=O\bigl(t^2L^3/\epsilon\bigr)$. Therefore, the randomized first-order product formula algorithm improves over the deterministic algorithm with respect to all parameters of interest.

It is natural to ask whether a similar randomization strategy can improve higher-order product formulas (as defined in \eq{pf2k}).
While it turns out that randomization does not improve the order of the formula, it does result in a significant reduction of the error, and in particular, lowers the dependence on the number of summands in the Hamiltonian. The more complicated structure of higher-order formulas makes this analysis more involved than in the first-order case (in particular, we randomly permute the $L$ summands instead of simply choosing whether or not to reverse them, so we use $\Theta(L \log L)$ bits of randomness per segment instead of only a single bit). As discussed at the end of \sec{intro}, our proof is based on a randomization lemma (established in the next section) that evaluates the dominant contribution to the Taylor series of the randomized product formula in closed form.

\section{Randomization lemma}
\label{sec:rand}

In this section, we study the Taylor expansion of the average evolution operator obtained by randomizing how the summands of a Hamiltonian are ordered. We consider a formula of the form
\begin{equation}
\begin{aligned}
&\exp(q_1\lambda H_{\pi_1(1)})\exp(q_1\lambda H_{\pi_1(2)})\cdots\exp(q_1\lambda H_{\pi_1(L)})\\
&\exp(q_2\lambda H_{\pi_2(1)})\exp(q_2\lambda H_{\pi_2(2)})\cdots\exp(q_2\lambda H_{\pi_2(L)})\\
&\cdots\\
&\exp(q_\kappa\lambda H_{\pi_\kappa(1)})\exp(q_\kappa\lambda H_{\pi_\kappa(2)})\cdots\exp(q_\kappa\lambda H_{\pi_\kappa(L)})
\end{aligned}
\end{equation}
for real numbers $q_1,\ldots,q_\kappa\in\R$, a complex number $\lambda\in\C$, Hermitian matrices $H_1,\ldots,H_L$, and permutations $\pi_1,\ldots,\pi_\kappa\in\Sym(L)$.
By choosing appropriate values of $q_1,\ldots,q_\kappa\in\R$ and ordering $H_1,\ldots,H_L$ in both forward and backward directions, we can write any product formula $S_{2k}(\lambda)$ in this form.

We now permute the summands to get the average evolution
\begin{equation}
\label{eq:avg}
\begin{aligned}
\frac{1}{L!}\sum_{\sigma\in \Sym(L)}&\exp(q_1\lambda H_{\sigma(\pi_1(1))})\exp(q_1\lambda H_{\sigma(\pi_1(2))})\cdots\exp(q_1\lambda H_{\sigma(\pi_1(L))})\\
&\exp(q_2\lambda H_{\sigma(\pi_2(1))})\exp(q_2\lambda H_{\sigma(\pi_2(2))})\cdots\exp(q_2\lambda H_{\sigma(\pi_2(L))})\\
&\cdots\\
&\exp(q_\kappa\lambda H_{\sigma(\pi_\kappa(1))})\exp(q_\kappa\lambda H_{\sigma(\pi_\kappa(2))})\cdots\exp(q_\kappa\lambda H_{\sigma(\pi_\kappa(L))}).
\end{aligned}
\end{equation}
In its Taylor expansion, we call the sum of the form
\begin{equation}
\sum_{\substack{m_1,\ldots,m_s\\ \text{pairwise different}}}\alpha_{m_1\ldots m_s}\lambda^s H_{m_1}\cdots H_{m_s},
\end{equation}
with coefficients $\alpha_{m_1\ldots m_s}\in\C$, the $s$th-order \emph{nondegenerate term}. This term contributes $\Theta(L^s)$ to the $s$th-order error, whereas the remaining (degenerate) terms only contribute $O(L^{s-1})$.

The following lemma shows how to compute the $s$th-order nondegenerate term for an arbitrary average evolution.

\begin{lemma}[Randomization lemma]
	\label{lem:randlemma}
	Define an average evolution operator as in \eq{avg} and let $s\leq L$ be a positive integer. The $s$th-order nondegenerate term of this operator is
	\begin{equation}
	\label{eq:randlemmastate}
	\frac{[(q_1+\cdots+q_\kappa)\lambda]^s}{s!}\sum_{\substack{m_1,\ldots,m_s\\ \text{pairwise different}}}H_{m_1}\cdots H_{m_s}.
	\end{equation}
\end{lemma}

\begin{proof}
	We take all possible products of $s$ terms from the Taylor expansion of \eq{avg}. Observe that the exponentials in \eq{avg} are organized in an array with $\kappa$ rows and $L$ columns. We use $\kappa_1,\ldots,\kappa_s$ and $l_1,\ldots,l_s$ to label the row and column indices, respectively, of the exponentials from which the terms are chosen. To avoid double counting, we take terms with smaller row indices first (i.e., $\kappa_1\leq\cdots\leq \kappa_s$). Within each row, we take terms with smaller column indices first. To get the $s$th-order nondegenerate term, we require that $\pi_{\kappa_1}(l_1),\ldots,\pi_{\kappa_s}(l_s)$ are pairwise different. The $s$th-order nondegenerate term of \eq{avg} can then be expressed as
	\begin{align}
	\frac{1}{L!}\sum_{\sigma\in \Sym(L)}\sum_{\kappa_1\leq\cdots\leq \kappa_s}\sum_{\substack{\pi_{\kappa_1}(l_1),\ldots,\pi_{\kappa_s}(l_s)\\ \text{pairwise different}}}
	(q_{\kappa_1}\lambda H_{\sigma(\pi_{\kappa_1}(l_1))})\cdots(q_{\kappa_s}\lambda H_{\sigma(\pi_{\kappa_s}(l_s))}).
	\end{align}
	
	A direct calculation shows that
	\begin{equation}
	\begin{aligned}
	&\frac{1}{L!}\sum_{\sigma\in \Sym(L)}\ \sum_{\kappa_1\leq\cdots\leq \kappa_s}\ \sum_{\substack{\pi_{\kappa_1}(l_1),\ldots,\pi_{\kappa_s}(l_s)\\ \text{pairwise different}}}
	(q_{\kappa_1}\lambda H_{\sigma(\pi_{\kappa_1}(l_1))})\cdots(q_{\kappa_s}\lambda H_{\sigma(\pi_{\kappa_s}(l_s))}) \\
	&\quad=\frac{1}{L!}\sum_{\sigma\in \Sym(L)}\ \sum_{\kappa_1\leq\cdots\leq \kappa_s}\ \sum_{\substack{\pi_{\kappa_1}(l_1),\ldots,\pi_{\kappa_s}(l_s)\\ \text{pairwise different}}}\ \sum_{\substack{m_1=\sigma(\pi_{\kappa_1}(l_1)),\ldots,\\m_s=\sigma(\pi_{\kappa_s}(l_s))}}
	(q_{\kappa_1}\lambda H_{m_1})\cdots(q_{\kappa_s}\lambda H_{m_s}) \\
	&\quad=\frac{1}{L!}\sum_{\substack{m_1,\ldots,m_s\\ \text{pairwise different}}}\ \sum_{\kappa_1\leq\cdots\leq \kappa_s}\ \sum_{\substack{\pi_{\kappa_1}(l_1),\ldots,\pi_{\kappa_s}(l_s)\\ \text{pairwise different}}}\ 
	\sum_{\substack{\sigma\in \Sym(L):\\ \sigma(\pi_{\kappa_1}(l_1))=m_1,\ldots,\\\sigma(\pi_{\kappa_s}(l_s))=m_s}}
	(q_{\kappa_1}\lambda H_{m_1})\cdots(q_{\kappa_s}\lambda H_{m_s}) \\
	&\quad=\frac{(L-s)!}{L!}\sum_{\substack{m_1,\ldots,m_s\\ \text{pairwise different}}}\bigg[\sum_{\kappa_1\leq\cdots\leq \kappa_s}\ \sum_{\substack{\pi_{\kappa_1}(l_1),\ldots,\pi_{\kappa_s}(l_s)\\ \text{pairwise different}}}
	(q_{\kappa_1}\lambda )\cdots(q_{\kappa_s}\lambda)\bigg]H_{m_1}\cdots H_{m_s}.
	\label{eq:randlemmaproof}
	\end{aligned}
	\end{equation}
	Now observe that the summand $(q_{\kappa_1}\lambda )\cdots(q_{\kappa_s}\lambda)$ depends only on the row indices. Letting $r_1,\ldots,r_\kappa$ denote the number of terms picked from row $1,\ldots,\kappa$, respectively, we can re-express this summand as $(q_{1}\lambda)^{r_1}\cdots(q_\kappa\lambda)^{r_\kappa}$. We determine the coefficient of this term as follows. The number of ways of choosing $l_1,\ldots,l_s$ pairwise different is $L(L-1)\cdots(L-s+1)$. However, when we apply permutations $\pi_{\kappa_1},\ldots,\pi_{\kappa_s}$, we may double count some terms. In particular, if $\kappa_i=\kappa_{i+1}$, we are to pick terms from the same row $\kappa_i$ and we must have $l_i<l_{i+1}$. This implies that the ordering of $\pi_{\kappa_i}(l_i)$ and $\pi_{\kappa_{i+1}}(l_{i+1})$ is uniquely determined. Altogether, we see that we have overcounted by a factor of $(r_1!)\cdots (r_\kappa!)$. Therefore, we have
	\begin{equation}
	\label{eq:multinomial}
	\begin{aligned}
	\sum_{\kappa_1\leq\cdots\leq \kappa_s}\sum_{\substack{\pi_{\kappa_1}(l_1),\ldots,\pi_{\kappa_s}(l_s)\\ \text{pairwise different}}} \!
	(q_{\kappa_1}\lambda )\cdots(q_{\kappa_s}\lambda)
	&= \! \sum_{\substack{r_1,\ldots, r_\kappa:\\ r_1+\cdots+r_\kappa=s}} \!\!\!
	\frac{L(L-1)\cdots(L-s+1)}{(r_1!)\cdots (r_\kappa!)}(q_{1}\lambda)^{r_1}\cdots(q_\kappa\lambda)^{r_\kappa}\\
	&= L(L-1)\cdots(L-s+1)\frac{[(q_1+\cdots+q_\kappa)\lambda]^s}{s!},
	\end{aligned}
	\end{equation}
	where the last equality follows by the multinomial theorem.
	
	Substituting \eq{multinomial} into \eq{randlemmaproof} completes the proof.
\end{proof}

As an immediate corollary, we compute the $s$th-order nondegenerate term of the average evolution operator $\frac{1}{L!}\sum_{\sigma\in \Sym(L)}S_{2k}^\sigma(\lambda)$.
\begin{corollary}
	\label{cor:randpf}
	Let $\{H_j\}_{j=1}^{L}$ be Hermitian operators; let $\lambda\in\C$, $k,s\in\N$, and $s\leq L$. Then the $s$th-order nondegenerate term of the average evolution $\frac{1}{L!}\sum_{\sigma\in \Sym(L)}S_{2k}^\sigma(\lambda)$, with $S_{2k}^\sigma(\lambda)$ defined in \eq{permpf}, is
	\begin{equation}
	\label{eq:nondegpf}
	\frac{\lambda^s}{s!}\sum_{\substack{m_1,\ldots,m_s\\ \text{pairwise different}}}H_{m_1}\cdots H_{m_s}.
	\end{equation}
\end{corollary}
\begin{proof}
	The fact that $S_{2k}^\sigma(\lambda)$ is at least first-order accurate implies that $q_1+\cdots+ q_\kappa=1$ in \eq{randlemmastate}.
\end{proof}
Observe that the $s$th-order nondegenerate term of $V(\lambda)=\exp\bigl(\lambda\sum_{j=1}^{L}H_j\bigr)$ is also given by \eq{nondegpf}. Therefore, the $s$th-order nondegenerate term completely cancels in
\begin{equation}
V(\lambda)-\frac{1}{L!}\sum_{\sigma\in \Sym(L)}S_{2k}^\sigma(\lambda).
\end{equation}

\section{Error bounds}
\label{sec:errorbound}

In this section we establish our main result, an upper bound on the error of a randomized product formula simulation. To apply the mixing lemma, we need to bound the error of the average evolution. We now present an error bound for an arbitrary fixed-order term in the Taylor expansion of the average evolution operator.

\begin{lemma}
	\label{lem:bterm}
	Let $\{H_j\}_{j=1}^{L}$ be Hermitian operators; let $\lambda\in\C$ and $k,s\in\N$. Define the target evolution $V(\lambda)$ as in \eq{idealexp}, and define the permuted $(2k)$th-order formula $S_{2k}^\sigma(\lambda)$ as in \eq{permpf}. Then the $s$th-order error of the approximation
	\begin{equation}
	\label{eq:rand_fix}
	V(\lambda)-\frac{1}{L!}\sum_{\sigma\in \Sym(L)}S_{2k}^\sigma(\lambda)
	\end{equation}
	is at most
	\begin{equation}
	\begin{cases}
	0 & 0\leq s\leq 2k,\\
	\frac{(2\cdot 5^{k-1}\Lambda|\lambda|)^s}{(s-2)!}L^{s-1} & s>2k,
	\end{cases}
	\end{equation}
	where $\Lambda:=\max\norm{H_j}$.
\end{lemma}

The proof of this error bound uses the following estimate of a fixed-order degenerate term in the average evolution operator.

\begin{lemma}
	\label{lem:degbound}
	Let $\{H_j\}_{j=1}^{L}$ be Hermitian operators with $\Lambda := \max_{j}\norm{H_{j}}$; let $q_1,\ldots,q_\kappa \in \R$ with $\max_k |q_k| \le 1$; and let $s \le L$ be a positive integer.
	Then the norm of the $s$th-order degenerate term of the ideal evolution operator $V(\lambda)$ as in \eq{idealexp} is at most
	\begin{equation}
	\label{eq:degideal}
		\frac{(\Lambda|\lambda|)^s}{s!}\big[L^s-L(L-1)\cdots(L-s+1)\big]
	\end{equation}
	and the norm of the $s$th-order degenerate term of the average evolution operator as in \eq{avg} is at most
	\begin{equation}
	\label{eq:degavg}
		\frac{(\kappa\Lambda|\lambda|)^s}{s!}\big[L^s-L(L-1)\cdots(L-s+1)\big].
	\end{equation}
\end{lemma}
\begin{proof}
	The $s$th-order term of $V(\lambda)$ is
	\begin{equation}
		\frac{\big(\lambda\sum_{j=1}^{L}H_j\big)^s}{s!}=\frac{\lambda^s}{s!}\sum_{\substack{m_1,\ldots,m_s}}H_{m_1}\cdots H_{m_s}
	\end{equation}
	and its nondegenerate term is
	\begin{equation}
		\frac{\lambda^s}{s!}\sum_{\substack{m_1,\ldots,m_s\\ \text{pairwise different}}}H_{m_1}\cdots H_{m_s}.
	\end{equation}
	We use the following strategy to bound the norms of these terms: (i) bound the norm of a sum of terms by summing the norms of each term; (ii) bound the norm of a product of terms by multiplying the norms of each term; (iii) bound the norm of each summand by $\Lambda$; and (iv) replace $\lambda$ by $|\lambda|$. Applying this strategy, we find that the norm of the $s$th-order term is at most $(L\Lambda|\lambda|)^s/s!$, where the nondegenerate term contributes precisely $L(L-1)\cdots(L-s+1)(\Lambda|\lambda|)^s/s!$. Taking the difference gives the desired bound \eq{degideal}.

  According to \lem{randlemma}, the $s$th-order nondegenerate term of the average evolution is
	\begin{equation}
		\frac{[(q_1+\cdots+q_\kappa)\lambda]^s}{s!}\sum_{\substack{m_1,\ldots,m_s\\ \text{pairwise different}}}H_{m_1}\cdots H_{m_s}.
	\end{equation}
	Following the same strategy as for $V(\lambda)$ and also upper bounding the norm of each $q_k$ by $1$ as part of step (iv), we find that the norm of this term is at most
	\begin{equation}
		\frac{(\kappa\Lambda|\lambda|)^s}{s!}L(L-1)\cdots(L-s+1).
		\label{eq:nondegpart}
	\end{equation}

	It remains to find an upper bound for the entire $s$th-order term of the average evolution. To this end, we start with the average evolution \eq{avg} and apply the following strategy: (i${}'$) replace each summand of the Hamiltonian by $\Lambda$; (ii${}'$) replace each $q_k$ by $1$ and each $\lambda$ by $|\lambda|$; and (iii${}'$) expand all exponentials into their Taylor series and extract the $s$th-order term. In other words, we extract the $s$th-order term of $\sum_{\sigma\in \Sym(L)}\exp(\kappa L \Lambda |\lambda|)/L!$ to get
	\begin{equation}
		\frac{(\kappa L\Lambda|\lambda|)^s}{s!}.
		\label{eq:entire}
	\end{equation}
	The equivalence of strategies (i)--(iv) and (i${}'$)--(iii${}'$) can be seen from \cite[Eq.~(57)]{CMNRS18}. Finally, taking the difference between \eq{entire} and \eq{nondegpart} gives the desired bound \eq{degavg}.
\end{proof}

\begin{proof}[Proof of \lem{bterm}]
	We first prove a stronger bound, namely that the $s$th-order error is at most
	\begin{equation}
	\label{eq:stronger}
	\begin{cases}
	0 & 0\leq s\leq 2k,\\
	2\frac{(2\cdot 5^{k-1}\Lambda|\lambda|)^s}{s!}[L^s-L(L-1)\cdots(L-s+1)] & 2k< s\leq L,\\
	2\frac{(2\cdot 5^{k-1}\Lambda|\lambda|)^s}{s!}L^s & s>L.
	\end{cases}
	\end{equation}
	The first and third cases in this expression are straightforward. The formula $S_{2k}^\sigma$ is exact for terms with order $0\leq s\leq 2k$ (this is what it means for the formula to have order $2k$), so the error is zero in this case. When $s>L$, the randomization lemma is not applicable and the error can be bounded as in \cite[Proof of Proposition F.3]{CMNRS18}.

	To handle the remaining case $2k< s\leq L$, we apply \lem{degbound} with $\kappa=2\cdot 5^{k-1}$. This choice of $\kappa$ follows from the definition of the ($2k$)th-order formula \eq{pf2k}. The norm of the $s$th-order degenerate terms can be upper bounded by
	\begin{equation}
		\frac{(\Lambda|\lambda|)^s}{s!}\big[L^s-L(L-1)\cdots(L-s+1)\big]+\frac{(2\cdot 5^{k-1}\Lambda|\lambda|)^s}{s!}\big[L^s-L(L-1)\cdots(L-s+1)\big].
	\end{equation}
	According to \cor{randpf}, the $s$th-order nondegenerate term of \eq{rand_fix} cancels, which proves \eq{stronger} for $2k< s\leq L$.

	To finish the proof, we need a unified error expression for order $s> 2k$. When $2k< s\leq L$, we have
	\begin{equation}
	\begin{aligned}
	&L^s-L(L-1)\cdots(L-s+1)\\
	&\quad=\#\big\{(l_1,\ldots,l_s)\in[L]^s\big\}-\#\big\{(l_1,\ldots,l_s)\in[L]^s: \forall i, j,\ l_i\neq l_j\big\}\\
	&\quad=\#\big\{(l_1,\ldots,l_s)\in[L]^s\big\}-\#\bigcap_{i< j}\big\{(l_1,\ldots,l_s)\in[L]^s:\ l_i\neq l_j\big\}\\
	&\quad=\#\bigcup_{i< j}\big\{(l_1,\ldots,l_s)\in[L]^s:\ l_i= l_j\big\}\\
	&\quad\leq\binom{s}{2}L^{s-1},
	\end{aligned}
	\end{equation}
	with $\#\{\cdot\}$ denoting the size of a set and $[L]:=\{1,\ldots,L\}$, where the inequality follows from the union bound. Therefore, we have
	\begin{equation}
		\begin{aligned}
			2\frac{(2\cdot 5^{k-1}\Lambda|\lambda|)^s}{s!}[L^s-L(L-1)\cdots(L-s+1)]
			&\leq \frac{(2\cdot 5^{k-1}\Lambda|\lambda|)^s}{s!} s(s-1)L^{s-1} \\
			&=\frac{(2\cdot 5^{k-1}\Lambda|\lambda|)^s }{(s-2)!}L^{s-1}.
		\end{aligned}
	\end{equation}

	If $s>L\in\N$, we have $s(s-1)\geq (L+1)L\geq 2L$ and
	\begin{equation}
	2\frac{(2\cdot 5^{k-1}\Lambda|\lambda|)^s}{s!}L^s
	\leq \frac{(2\cdot 5^{k-1}\Lambda|\lambda|)^s }{(s-2)!}L^{s-1}.
	\end{equation}
	This completes the proof.
\end{proof}


We also use the following standard tail bound on the exponential function \cite[Lemma F.2]{CMNRS18}.

\begin{lemma}
	\label{lem:tail}
	For any $x\in\C$ and $\kappa\in\N$, we have
	\begin{equation}
	\bigg|\sum_{s=\kappa}^{\infty}\frac{x^s}{s!}\bigg|\leq\frac{|x|^{\kappa}}{\kappa!}\exp(|x|).
	\end{equation}
\end{lemma}

We now establish the main theorem, which upper bounds the error of a higher-order randomized product formula.

\begin{theorem}[Randomized higher-order error bound]
	\label{thm:main}
	Let $\{H_j\}_{j=1}^{L}$ be Hermitian matrices. Let
	\begin{equation}
	V(-it):=\exp\biggl(-it\sum_{j=1}^{L}H_j\biggr)
	\end{equation}
	be the evolution induced by the Hamiltonian $H=\sum_{j=1}^{L}H_j$ for time $t$. For any permutation $\sigma\in\Sym(L)$, define the permuted $(2k)$th-order formula recursively by
	\begin{equation}
	\label{eq:permpf2}
	\begin{aligned}
	S_2^\sigma(\lambda)&:=\prod_{j=1}^{L}\exp\bigg(\frac{\lambda}{2}H_{\sigma(j)}\bigg)\prod_{j=L}^{1}\exp\bigg(\frac{\lambda}{2}H_{\sigma(j)}\bigg)\\
	S_{2k}^\sigma(\lambda)&:=[S_{2k-2}^\sigma(p_{k}\lambda)]^{2}S_{2k-2}^\sigma((1-4p_{k})\lambda)[S_{2k-2}^\sigma(p_{k}\lambda)]^{2},
	\end{aligned}
	\end{equation}
	with $p_{k}:=1/(4-4^{1/(2k-1)})$ for $k>1$. Let $r\in\N$ and $\Lambda:=\max\norm{H_j}$. Then
	\begin{equation}
	\begin{aligned}
	&\Biggl\|\mathcal{V}(-it)-\biggl(\frac{1}{L!}\sum_{\sigma\in\Sym(L)}\mathcal{S}_{2k}^\sigma(-it/r)\biggr)^r\Biggr\|_\diamond\\
	&\quad \leq 4\frac{(2\cdot 5^{k-1}\Lambda |t| L)^{4k+2}}{\bigl((2k+1)!\bigr)^2 r^{4k+1}}\exp\biggl(4\cdot 5^{k-1}\frac{\Lambda |t| L}{r}\biggr)+
	2\frac{(2\cdot 5^{k-1}\Lambda |t|)^{2k+1}L^{2k}}{(2k-1)! r^{2k}}\exp\biggl(2\cdot 5^{k-1}\frac{\Lambda |t| L}{r}\biggr)
	\end{aligned}
	\end{equation}
	where, for $\lambda=-it$, we associate quantum channels $\mathcal{V}(\lambda)$ and $\mathcal{S}_{2k}^\sigma(\lambda)$ with the unitaries $V(\lambda)$ and $S_{2k}^\sigma(\lambda)$, respectively.
\end{theorem}
\begin{proof}
	We first prove that
	\begin{equation}
	\label{eq:oneseg}
	\begin{aligned}
	&\norm{\mathcal{V}(\lambda)-\frac{1}{L!}\sum_{\sigma\in\Sym(L)}\mathcal{S}_{2k}^\sigma(\lambda)}_\diamond\\
	&\quad \leq 4\frac{(2\cdot 5^{k-1}\Lambda |\lambda|L)^{4k+2}}{\bigl((2k+1)!\bigr)^2}\exp\bigl(4\cdot 5^{k-1}\Lambda |\lambda|L\bigr)+
	2\frac{(2\cdot 5^{k-1}\Lambda |\lambda|)^{2k+1}L^{2k}}{(2k-1)!}\exp\bigl(2\cdot 5^{k-1}\Lambda |\lambda|L\bigr).
	\end{aligned}
	\end{equation}
	To this end, note that the $s$th-order error of $V(\lambda)-S_{2k}^\sigma(\lambda)$ is at most
	\begin{equation}
	\begin{cases}
	0 & 0\leq s\leq 2k,\\
	\frac{2(2\cdot 5^{k-1}\Lambda|\lambda|)^s}{s!}L^{s} & s>2k
	\end{cases}
	\end{equation}
	(as before, this follows as in \cite[Proof of Proposition F.3]{CMNRS18}).
	Thus \lem{tail} gives
	\begin{equation}
	\norm{V(\lambda)-S_{2k}^\sigma(\lambda)}\leq
	2\frac{(2\cdot 5^{k-1}\Lambda |\lambda|L)^{2k+1}}{(2k+1)!}\exp\bigl(2\cdot 5^{k-1}\Lambda |\lambda|L\bigr).
	\end{equation}
	On the other hand, \lem{bterm} implies that the $s$th-order error of $V(\lambda)-\frac{1}{L!}\sum_{\sigma\in \Sym(L)}S_{2k}^\sigma(\lambda)$ is at most
	\begin{equation}
	\begin{cases}
	0 & 0\leq s\leq 2k,\\
	\frac{(2\cdot 5^{k-1}\Lambda|\lambda|)^s}{(s-2)!}L^{s-1} & s>2k,
	\end{cases}
	\end{equation}
	so again \lem{tail} gives
	\begin{equation}
	\norm{V(\lambda)-\frac{1}{L!}\sum_{\sigma\in \Sym(L)}S_{2k}^\sigma(\lambda)}\leq
	\frac{(2\cdot 5^{k-1}\Lambda |\lambda|)^{2k+1}L^{2k}}{(2k-1)!}\exp\bigl(2\cdot 5^{k-1}\Lambda |\lambda|L\bigr).
	\end{equation}
	Equation~\eq{oneseg} now follows from \lem{mix} by setting
	\begin{equation}
	\begin{aligned}
	a&=2\frac{(2\cdot 5^{k-1}\Lambda |\lambda|L)^{2k+1}}{(2k+1)!}\exp\bigl(2\cdot 5^{k-1}\Lambda |\lambda|L\bigr),\\
	b&=\frac{(2\cdot 5^{k-1}\Lambda |\lambda|)^{2k+1}L^{2k}}{(2k-1)!}\exp\bigl(2\cdot 5^{k-1}\Lambda |\lambda|L\bigr).
	\end{aligned}
	\end{equation}
	
	To simulate the evolution for time $t$, we divide it into $r$ segments. The error within each segment is obtained from \eq{oneseg} by setting $\lambda=-it/r$. Then subadditivity of the diamond norm distance gives
	\begin{equation}
	\norm{\mathcal{V}(-it)-\bigg(\frac{1}{L!}\sum_{\sigma\in\Sym(L)}\mathcal{S}_{2k}^\sigma\big(-it/r\big)\bigg)^r}_\diamond
	\leq r\norm{\mathcal{V}(-it/r)-\frac{1}{L!}\sum_{\sigma\in\Sym(L)}\mathcal{S}_{2k}^\sigma\big(-it/r\big)}_\diamond,
	\end{equation}
	which completes the proof.
\end{proof}

\section{Algorithm performance and comparisons}
\label{sec:performance}

We now analyze the complexity of our randomized product formula algorithm. Assume that $k\in\N$ is fixed, $\Lambda=O(1)$ is constant, and $r>tL$. By \thm{main}, the asymptotic error of the $(2k)$th-order randomized product formula is
\begin{equation}
\norm{\mathcal{V}(-it)-\bigg(\frac{1}{L!}\sum_{\sigma\in\Sym(L)}\mathcal{S}_{2k}^\sigma\big(-it/r\big)\bigg)^r}_\diamond
\leq O\bigg(\frac{(tL)^{4k+2}}{r^{4k+1}}+\frac{t^{2k+1}L^{2k}}{r^{2k}}\bigg).
\label{eq:randpferror}
\end{equation}
To guarantee that the simulation error is at most $\epsilon$, we upper bound the right-hand side of \eq{randpferror} by $\epsilon$ and solve for $r$. We find that it suffices to use
\begin{equation}
\begin{aligned}
r_{2k}^\rand
&=\max\bigg\{O\bigg(\frac{(tL)^{\frac{4k+2}{4k+1}}}{\epsilon^{\frac{1}{4k+1}}}\bigg),O\bigg(\frac{t^{\frac{2k+1}{2k}}L}{\epsilon^{\frac{1}{2k}}}\bigg)\bigg\}\\
&=\max\bigg\{O\bigg(tL\bigg(\frac{tL}{\epsilon}\bigg)^{\frac{1}{4k+1}}\bigg),O\bigg(tL\bigg(\frac{t}{\epsilon}\bigg)^{\frac{1}{2k}}\bigg)\bigg\}
\end{aligned}
\end{equation}
segments, giving a simulation algorithm with
\begin{equation}
g_{2k}^\rand
=O(L r_{2k}^\rand)
=\max\bigg\{O\bigg(tL^2\bigg(\frac{tL}{\epsilon}\bigg)^{\frac{1}{4k+1}}\bigg),
O\bigg(tL^2\bigg(\frac{t}{\epsilon}\bigg)^{\frac{1}{2k}}\bigg)\bigg\}
\label{eq:randpfgates}
\end{equation}
elementary gates.

For comparison, the error in the ($2k$)th-order deterministic formula algorithm is at most \cite[Proposition F.4]{CMNRS18}
\begin{equation}
\norm{V(-it)-\bigl[S_{2k}(-it/r)\bigr]^r}\leq
O\biggl(\frac{(tL)^{2k+1}}{r^{2k}}\biggr).
\end{equation}
While this bound quantifies the simulation error in terms of the spectral-norm distance, it can easily be adapted to the diamond-norm distance using either \lem{mix} or \cite[Lemma 7]{BCK15}. This translation introduces only constant-factor overhead, so we have
\begin{equation}
\norm{\mathcal{V}(-it)-\bigl[\mathcal{S}_{2k}(-it/r)\bigr]^r}_\diamond\leq
O\biggl(\frac{(tL)^{2k+1}}{r^{2k}}\biggr).
\end{equation}
Therefore, the number of segments that suffice to ensure error at most $\epsilon$ satisfies
\begin{equation}
r_{2k}^{\det} = O\biggl(tL\biggl(\frac{tL}{\epsilon}\biggr)^{\frac{1}{2k}}\biggr),
\end{equation}
giving an algorithm with
\begin{equation}
g_{2k}^{\det} = O(L r_{2k}^{\det}) =  O\biggl(tL^2\biggl(\frac{tL}{\epsilon}\biggr)^{\frac{1}{2k}}\biggr)
\end{equation}
elementary gates. Comparing to \eq{randpfgates}, we see that the randomized product formula strictly improves the complexity as a function of $L$. Indeed, the $(2k)$th-order randomized approach either provides an improvement with respect to all parameters of interest over the $(2k)$th order deterministic approach (if the first term of \eq{randpfgates} obtains the maximum), or has better dependence on the number of terms in the Hamiltonian than any deterministic formula (if the second term dominates).

We can also compare our result to the commutator bound of \cite{CMNRS18}, which depends on the specific structure of the Hamiltonian. For concreteness, we consider a one-dimensional nearest-neighbor Heisenberg model with a random magnetic field, as studied in \cite{CMNRS18}. Specifically, let
\begin{align}
H=\sum_{j=1}^n (\vec \sigma_j \cdot \vec \sigma_{j+1} + h_j \sigma_j^z)
\label{eq:heisenberg}
\end{align}
with periodic boundary conditions (i.e., $\vec\sigma_{n+1} = \vec\sigma_1$), and $h_j \in [-h,h]$ chosen uniformly at random, where $\vec \sigma_j = (\sigma^x_j,\sigma^y_j,\sigma^z_j)$ denotes a vector of Pauli $x$, $y$, and $z$ matrices on qubit $j$. The $(2k)$th-order deterministic formula with the commutator bound has error at most \cite[Eq.~(146)]{CMNRS18}
\begin{equation}
\norm{\mathcal{V}(-it)-\big[\mathcal{S}_{2k}\big(-it/r\big)\big]^r}_\diamond
\leq O\biggl(\frac{(tL)^{2k+2}}{r^{2k+1}}+\frac{t^{2k+1}L^{2k}}{r^{2k}}\biggr),
\end{equation}
where we have again used \lem{mix} (or \cite[Lemma 7]{BCK15}) to relate the spectral-norm distance to the diamond-norm distance. To guarantee that the simulation error is at most $\epsilon$, it suffices to choose
\begin{equation}
\begin{aligned}
r_{2k}^\comm&=\max\biggl\{O\biggl(\frac{(tL)^{\frac{2k+2}{2k+1}}}{\epsilon^{\frac{1}{2k+1}}}\biggr),
O\biggl(\frac{t^{\frac{2k+1}{2k}}L}{\epsilon^{\frac{1}{2k}}}\biggr)\biggr\}\\
&=\max\biggl\{O\biggl(tL\biggl(\frac{tL}{\epsilon}\biggr)^{\frac{1}{2k+1}}\biggr),
O\biggl(tL\biggl(\frac{t}{\epsilon}\biggr)^{\frac{1}{2k}}\biggr)\biggr\}
\end{aligned}
\end{equation}
segments, giving an algorithm with
\begin{equation}
g_{2k}^\comm=O(L r_{2k}^\comm)
=\max\biggl\{O\biggl(tL^2\biggl(\frac{tL}{\epsilon}\biggr)^{\frac{1}{2k+1}}\biggr),
O\biggl(tL^2\biggl(\frac{t}{\epsilon}\biggr)^{\frac{1}{2k}}\biggr)\biggr\}
\label{eq:commpfgates}
\end{equation}
elementary gates. Comparing to the corresponding bound \eq{randpfgates} for randomized product formulas, we see that the only difference is that the exponent $1/(2k+1)$ for the commutator bound becomes $1/(4k+1)$ in the randomized case. Thus the randomized approach can provide a slightly faster algorithm despite using less information about the structure of the Hamiltonian. More specifically, the relationship between $t$ and $L$ determines whether the randomized approach offers an improvement. If $t = \Omega(L^{2k})$, then the second term of \eq{commpfgates} achieves the maximum, and both approaches have asymptotic complexity $O\big(tL^2\bigl(\frac{t}{\epsilon}\big)^{\frac{1}{2k}}\bigr)$. However, if $t = o(L^{2k})$, then the randomized formula is advantageous.

\section{Empirical performance}
\label{sec:numerics}

While randomization provides a useful theoretical handle for establishing better provable bounds, those bounds may still be far from tight. As described in \sec{intro}, our original motivation for considering randomization was the observation that product formulas appear to perform dramatically better in practice than the best available proven bounds would suggest. To investigate the empirical behavior of product formulas, we numerically evaluate their performance for simulations of the Heisenberg model \eq{heisenberg} with $t=n$ and $h=1$, targeting error $\epsilon=10^{-3}$, as previously considered in \cite{CMNRS18}. We collect data for the first-, fourth-, and sixth-order formulas as the latter two orders have the best performance in practice for small $n$ and the first-order formula offers a qualitatively better theoretical improvement.

For the deterministic formula, we order the operators of the Hamiltonian in the same way as \cite{CMNRS18}, namely
\begin{align}
\sigma_1^x \sigma_2^x, \ldots, \sigma_{n-1}^x \sigma_n^x, \sigma_n^x \sigma_1^x, \;
\sigma_1^y \sigma_2^y, \ldots, \sigma_{n-1}^y \sigma_n^y, \sigma_n^y \sigma_1^y, \;
\sigma_1^z \sigma_2^z, \ldots, \sigma_{n-1}^z \sigma_n^z, \sigma_n^z \sigma_1^z,\;
\sigma_1^z, \ldots, \sigma_n^z.
\label{eq:cycletermorder}
\end{align}
We compute the error in terms of the spectral-norm distance and convert it to the diamond-norm distance using Lemma 7 of \cite{BCK15} (i.e., we multiply by $2$). To analyze the randomized formula, we would like to numerically evaluate the diamond-norm distances
\begin{equation}
\norm{\mathcal{V}(-it)-\frac{1}{2^r}\bigl(\mathcal{S}_1(-it/r)+\mathcal{S}_1^\rev(-it/r)\bigr)^r}_\diamond
\end{equation}
and
\begin{equation}
\label{eq:diamond_err}
\norm{\mathcal{V}(-it)-\biggl(\frac{1}{L!}\sum_{\sigma\in\Sym(L)}\mathcal{S}_{2k}^\sigma\bigl(-it/r\bigr)\biggr)^r}_\diamond.
\end{equation}
While the diamond norm can be computed using a semidefinite program \cite{Wat13}, direct computation is prohibitive as the channel contains $(L!)^r$ Kraus operators. Instead, we use \lem{mix} to estimate the error. We randomly choose the ordering of the summands in each of the $r$ segments, exponentiate each individual operator, and construct a unitary operator by concatenating the exponentials according to the given product formula. We follow this procedure to obtain a Monte Carlo estimate of the average error
\begin{equation}
\norm{V(-it)-\frac{1}{M}\sum_{m=1}^{M}S_{2k}^{\sigma_{m,r}}\bigl(-it/r\bigr)\cdots S_{2k}^{\sigma_{m,1}}\bigl(-it/r\bigr)}
\end{equation}
for the ($2k$)th-order formula and similarly for the first-order case. Here, $M$ is the number of samples in the Monte Carlo estimation, which can be increased to get more accurate estimate. In practice, we find that it suffices to take only three samples, as the standard deviations are already negligibly small (about $10^{-5}$). We then invoke \lem{mix} to bound the diamond-norm error in \eq{diamond_err}. To the extent that the bound of \lem{mix} is loose, we expect the empirical performance to be better in practice.


\begin{figure}[t]
	\centering
	\begin{subfigure}{.33\linewidth}
		\resizebox{.95\textwidth}{!}{
			\begin{tikzpicture}
			\begin{axis}[
			log x ticks with fixed point,
			xtick={6,7,8,9,10},
			xmode=log,
			ymode=log,
			xmin = 5,
			xmax = 12,
			ymin=10^3,
			ymax=10^6,
			width=10cm,
			ymajorgrids=true,
			yminorgrids=true,
			legend style={at={(0.95,0.05)},anchor=south east},
			xlabel={$n$},
			ylabel={$r$}, ylabel near ticks
			]
			
			\addlegendimage{empty legend}
			\addlegendentry[yshift=0pt]{\hspace{-.6cm}\textbf{First order}}
			
			\addplot[only marks,color=blue, mark=o,error bars/.cd,y dir=both,y explicit] coordinates {
				(6,184195.200000000) +- (0,19482.4528666182)
				(7,203581.400000000) +- (0,48659.6127943904)
				(8,288667.600000000) +- (0,53462.8840486557)
				(9,418935.800000000) +- (0,108816.890459156)
				(10,484212.200000000) +- (0,91375.0444880877)
			};
			\addlegendentry{Deterministic}
			
			\addplot[only marks,color=red, mark=o] coordinates {
				(6,7762.80000000000) +- (0,521.047214751216)
				(7,9496.40000000000) +- (0,932.050588755782)
				(8,13617.4000000000) +- (0,787.384785222575)
				(9,15578.2000000000) +- (0,853.978454060757)
				(10,19125.8000000000) +- (0,632.350535699939)
			};
			\addlegendentry{Randomized}
			
			\addplot[
			color = black,
			mark = none
			]	coordinates {
				( 1, 4142.60599825257 )
				( 100, 56044889.2195794 )
			};
			
			\addplot[
			color = black,
			mark = none
			]	coordinates {
				( 1, 300.010102449959 )
				( 100, 1225551.83899852 )
			};
			\end{axis}
			\end{tikzpicture}
		}
	\end{subfigure}%
	\begin{subfigure}{.33\linewidth}
		\resizebox{.95\textwidth}{!}{
			\begin{tikzpicture}
			\begin{axis}[
			log x ticks with fixed point,
			xtick={6,7,8,9,10},
			xmode=log,
			ymode=log,
			xmin = 5,
			xmax = 12,
			ymin=10^1,
			ymax=10^3,
			width=10cm,
			ymajorgrids=true,
			yminorgrids=true,
			legend style={at={(0.95,0.05)},anchor=south east},
			xlabel={$n$},
			ylabel={$r$}, ylabel near ticks
			]
			
			\addlegendimage{empty legend}
			\addlegendentry[yshift=0pt]{\hspace{-.25cm}\textbf{Fourth order}}
			
			\addplot[only marks,color=blue, mark=square,error bars/.cd,y dir=both,y explicit] coordinates {
				(6,81.4000000000000) +- (0,2.19089023002066)
				(7,101.800000000000) +- (0,4.86826457785524)
				(8,123.200000000000) +- (0,2.77488738510232)
				(9,146.800000000000) +- (0,4.60434577328854)
				(10,173) +- (0,3.67423461417477)
			};
			\addlegendentry{Deterministic}
			
			\addplot[only marks,color=red, mark=square] coordinates {
				(6,72.6000000000000) +- (0,0.547722557505166)
				(7,88.6000000000000) +- (0,2.60768096208106)
				(8,108.600000000000) +- (0,1.51657508881031)
				(9,128) +- (0,1.73205080756888)
				(10,151.200000000000) +- (0,1.09544511501033)
			};
			\addlegendentry{Randomized}
			
			\addplot[
			color = black,
			mark = none
			]	coordinates {
				( 1, 5.821105514916931 )
				( 100, 5081.82649541254 )
			};
			
			\addplot[
			color = black,
			mark = none
			]	coordinates {
				( 1, 5.45762090439935 )
				( 100, 4114.18461500955 )
			};
			\end{axis}
			\end{tikzpicture}
		}
	\end{subfigure}
	\begin{subfigure}{.33\linewidth}
		\resizebox{.95\textwidth}{!}{
			\begin{tikzpicture}
			\begin{axis}[
			log x ticks with fixed point,
			xtick={6,7,8,9,10},
			ytick={10,100},
			xmode=log,
			ymode=log,
			xmin = 5,
			xmax = 12,
			ymin=10^1,
			ymax=10^2,
			width=10cm,
			ymajorgrids=true,
			yminorgrids=true,
			legend style={at={(0.95,0.05)},anchor=south east},
			xlabel={$n$},
			ylabel={$r$}, ylabel near ticks
			]
			
			\addlegendimage{empty legend}
			\addlegendentry[yshift=0pt]{\hspace{-.6cm}\textbf{Sixth order}}
			
			\addplot[only marks,color=blue, mark=hexagon] coordinates {
				(6,21.8000000000000) +- (0,1.09544511501033)
				(7,25.4000000000000) +- (0,0.894427190999916)
				(8,31.2000000000000) +- (0,1.09544511501033)
				(9,35) +- (0,1.41421356237310)
				(10,38.8000000000000) +- (0,0.447213595499958)
			};
			\addlegendentry{Deterministic}
			
			\addplot[only marks,color=red, mark=hexagon] coordinates {
				(6,22) +- (0,0)
				(7,26.6000000000000) +- (0,0.894427190999916)
				(8,30.8000000000000) +- (0,1.09544511501033)
				(9,35.2000000000000) +- (0,1.09544511501033)
				(10,39.8000000000000) +- (0,0.447213595499958)
			};
			\addlegendentry{Randomized}
			
			\addplot[
			color = black,
			mark = none
			]	coordinates {
				( 1, 2.71914388470343 )
				( 100, 568.523677947132 )
			};
			
			\addplot[
			color = black,
			mark = none
			]	coordinates {
				( 1, 2.80384591821328 )
				( 100, 2.80384591821328 )
			};
			\end{axis}
			\end{tikzpicture}
		}
	\end{subfigure}
	\caption{ Comparison of the values of $r$ between deterministic and randomized product formulas. Error bars are omitted when they are negligibly small on the plot. Straight lines show power-law fits to the data. \label{fig:extrapolate}}
\end{figure}
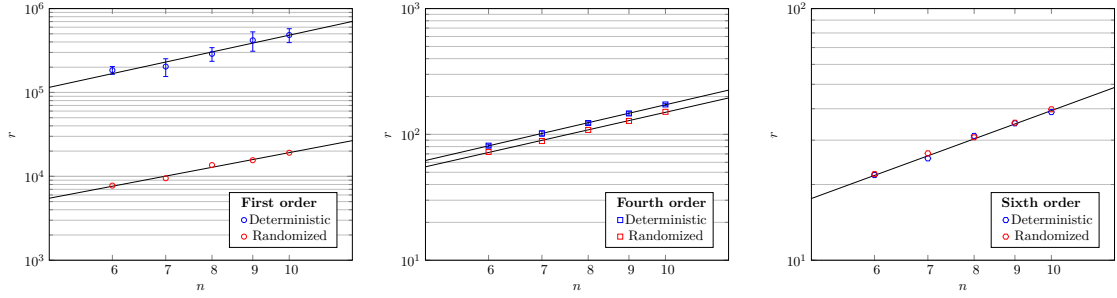


Using five randomly generated instances for each value of $n$, we apply binary search to determine the smallest number of segments $r$ that suffices to give error at most $10^{-3}$. \fig{extrapolate} shows the resulting data for the first-, fourth-, and sixth-order formulas, which are well-approximated by power laws. Fitting the data, we estimate that
\begin{align}
r_1^{\remp} &= 300.0 n^{1.806} &
r_4^{\remp} &= 5.458 n^{1.439} &
r_6^{\remp} &= 2.804 n^{1.152}
\end{align}
segments should suffice to give error at most $10^{-3}$. We thus observe that the empirical complexity of the randomized algorithm is still significantly better than the provable performance
\begin{align}
r_1^{\rand} &= O(n^{3}) &
r_4^{\rand} &= O(n^{2.25}) &
r_6^{\rand} &= O(n^{2.17}).
\end{align}
For comparison, analogous empirical fits for deterministic formulas give the comparable values
\begin{align}
r_1^{\demp} &= 4143 n^{2.066} &
r_4^{\demp} &= 5.821 n^{1.471} &
r_6^{\demp} &= 2.719 n^{1.160},
\end{align}
(cf.\ \cite[Eq.\ (147)]{CMNRS18}, but note that we have generated new data using \cite[Lemma 7]{BCK15} to bound the diamond-norm distance in terms of the spectral-norm distance),
whereas the rigorous commutator bound gives the larger exponents \cite{CMNRS18}
\begin{align}
r_1^{\comm} &= O(n^{3}) &
r_4^{\comm} &= O(n^{2.4}) &
r_6^{\comm} &= O(n^{2.28}).
\end{align}

We see that the randomized bound offers significantly better empirical performance at first order, consistent with the observation that randomization improves the order of approximation in this case. The fourth-order formula slightly improves both the exponent and the constant factor. While this improvement is small, it is nevertheless notable since it involves only a minor change to the algorithm. At sixth order we see negligible improvement. Since the proven bounds give less improvement with each successive order, it is perhaps not surprising to see that the empirical performance shows similar behavior.


\begin{figure}[t]
	\centering
	\resizebox{\textwidth}{!}{
		\begin{tikzpicture}
		\begin{axis}[
		log x ticks with fixed point,
		xtick={10,20,30,40,50,60,70,80,90,100},
		xmode=log,
		ymode=log,
		xmin=10,
		ymin=10^5,
		ymax=10^13,
		width=1.25\textwidth,
		height=1.\textwidth,
		ymajorgrids=true,
		yminorgrids=true,
		legend style={at={(0.015,0.985)},anchor=north west, font=\fontsize{8}{5}\selectfont},
		x label style={at={(axis description cs:0.5,-0.05)},anchor=north},
		y label style={at={(axis description cs:-0.03,.5)},anchor=south},
		xlabel={$n$}, xlabel near ticks,
		ylabel={number of exponentials}, ylabel near ticks,
		every axis legend/.append style={nodes={right},
			every x tick label/.append style={font=\small},
			every y tick label/.append style={font=\small},
		}
		]
		
		\addplot[only marks,color=blue, mark=square*] coordinates {
			(13,162839040)
			(16,336624000)
			(20,734698400)
			(25,1603596000)
			(32,3803065600)
			(40,8301592000)
			(50,18121970000)
			(63,40679150400)
			(79,89796790960)
			(100,204861400000)
		};
		\addlegendentry{Deterministic PF4 (min)}
		
		\addplot[only marks,color=blue, mark=square*,fill opacity=.2] coordinates {
			(13,14354600)
			(16,28790400)
			(20,60599200)
			(25,127149000)
			(32,287801600)
			(40,601040000)
			(50,1253376000)
			(63,2680274520)
			(79,5637762320)
			(100,12225628000)
		};
		\addlegendentry{Deterministic PF4 (com)}
		
		\addplot[only marks,color=red, mark=square*] coordinates {
			(13,51588680)
			(16,101219200)
			(20,208853600)
			(25,430967000)
			(32,960514560)
			(40,1982182400)
			(50,4090730000)
			(63,8663835600)
			(79,18066279320)
			(100,38842932000)
		};
		\addlegendentry{Randomized PF4 (min)}
		
		\addplot[only marks,color=blue, mark=square] coordinates {
			(13,138613.298664823)
			(16,236823.294609577)
			(20,421127.268682011)
			(25,748862.888340203)
			(32,1415666.67417323)
			(40,2517386.81721139)
			(50,4476503.19321865)
			(63,8125584.81484087)
			(79,14567816.5845061)
			(100,26759345.6037218)
		};
		\addlegendentry{Deterministic PF4 (emp)}
		
		\addplot[only marks,color=red, mark=square] coordinates {
			(13,115202.958105535)
			(16,192535.918371677)
			(20,334360.242688319)
			(25,580654.107743031)
			(32,1069293.40035558)
			(40,1856948.06388137)
			(50,3224798.83519913)
			(63,5711694.63843625)
			(79,9997067.82182107)
			(100,17909677.0156733)
		};
		\addlegendentry{Randomized PF4 (emp)}

		\addplot[only marks,color=blue, mark=hexagon*] coordinates {
			(13,1491510800)
			(16,2978048000)
			(20,6261556000)
			(25,13165960000)
			(32,29960537600)
			(40,63002464000)
			(50,132489490000)
			(63,286120220400)
			(79,608117478400)
			(100,1333687280000)
		};
		\addlegendentry{Deterministic PF6 (min)}
		
		\addplot[only marks,color=red, mark=hexagon*] coordinates {
			(13,427780600)
			(16,824691200)
			(20,1669800000)
			(25,3381065000)
			(32,7379513600)
			(40,14943536000)
			(50,30261770000)
			(63,62848384200)
			(79,128563904800)
			(100,270945760000)
		};
		\addlegendentry{Randomized PF6 (min)}
		
		\addplot[only marks,color=blue, mark=hexagon] coordinates {
			(13,138597.703306961)
			(16,217046.016593673)
			(20,351473.654799735)
			(25,569159.167060636)
			(32,970117.253271605)
			(40,1570960.21361210)
			(50,2543935.78140109)
			(63,4191045.71776868)
			(79,6833400.51482738)
			(100,11370473.5589426)
		};
		\addlegendentry{Deterministic PF6 (emp)}
		
		\addplot[only marks,color=red, mark=hexagon] coordinates {
			(13,140106.275931209)
			(16,219056.190973722)
			(20,354116.806008843)
			(25,572449.980712700)
			(32,973864.152151314)
			(40,1574306.85484576)
			(50,2544956.67361752)
			(63,4185235.65739779)
			(79,6811986.69442360)
			(100,11314184.1013371)
		};
		\addlegendentry{Randomized PF6 (emp)}
		\end{axis}
		\end{tikzpicture}
	}
	\caption{Comparison of the total number of elementary exponentials for product formula simulations of the Heisenberg model using deterministic and randomized product formulas of fourth and sixth order with both rigorous and empirical error bounds. Note that since the empirical performance of deterministic and randomized sixth-order product formulas is almost the same, the latter data points are obscured by the former.}
	\label{fig:compare}
\end{figure}
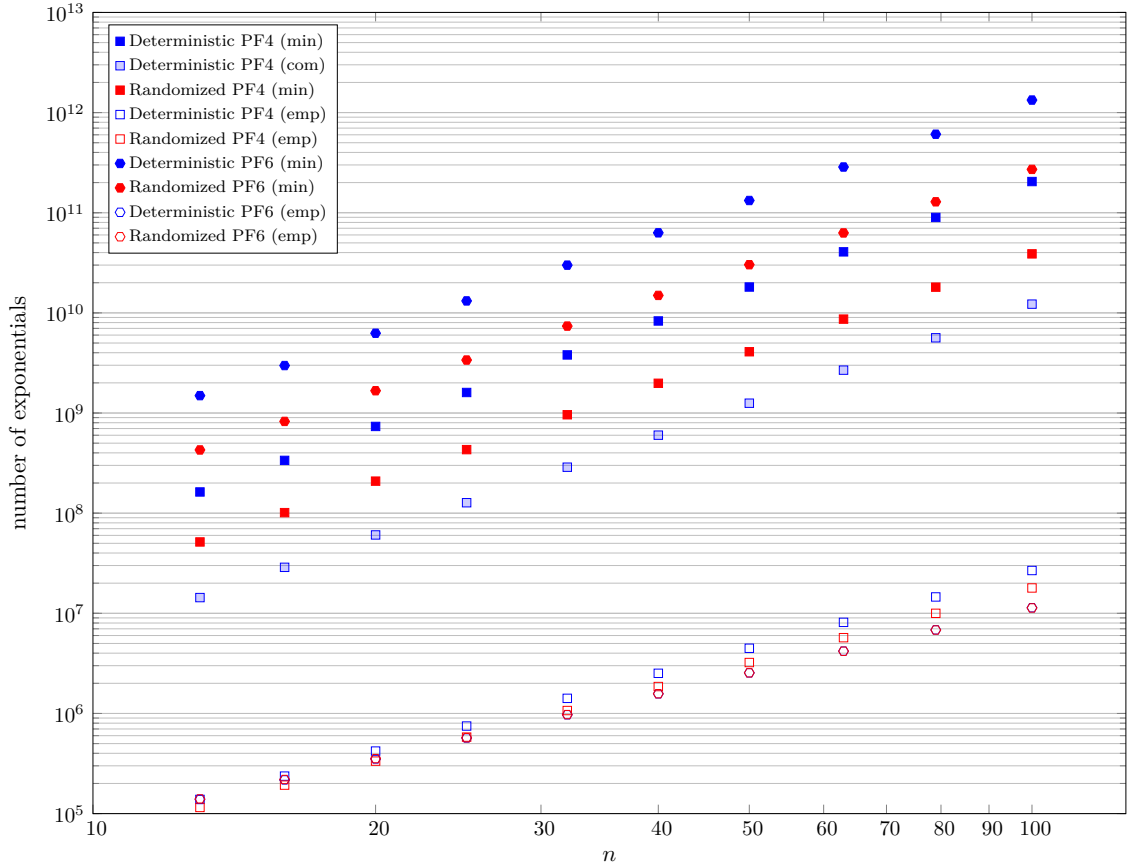


To illustrate the effect of using different formulas and different error bounds to simulate larger systems, \fig{compare} compares the cost of simulating our model system for sizes up to $n=100$ with deterministic and randomized formulas of orders $4$ and $6$, using both proven error bounds and the above empirical estimates. (We omit the first-order formula since it is not competitive even at such small sizes.) We give rigorous bounds for deterministic formulas using the minimized bound of \cite{CMNRS18}, and for fourth order we also show the result of using the commutator bound. We see that randomization gives a significant improvement over the deterministic formula using the minimized bound, although the commutator bound outperforms the randomized bound at the system sizes shown here. For sufficiently large $n$, the randomized bound gives lower complexity, but this requires a fairly large $n$ since the difference in exponents is small and the commutator bound achieves a favorable constant prefactor. Empirical estimates of the error improve the performance by several orders of magnitude, with randomization giving a small advantage for the fourth-order formula as indicated above. However, for systems of size larger than about $n=25$, the sixth-order bound prevails, and in this case randomization no longer offers a significant advantage.

\section{Discussion}
\label{sec:discussion}

We have shown that randomization can be used to establish better performance for quantum simulation algorithms based on product formulas. By simply randomizing how the summands in the Hamiltonian are ordered, we introduce terms in the average evolution that could not appear in any deterministic product formula approximation of the same order, and thereby give a more efficient algorithm. Indeed, this approach can outperform the commutator bound even though that method uses more information about the structure of the Hamiltonian. A randomized product formula simulation algorithm is not much more complicated than the corresponding deterministic formula, using only $O(L \log L)$ bits of randomness per segment and no ancilla qubits. Furthermore, we showed that randomization can even offer improved empirical performance in some cases.

While randomization has allowed us to make some progress on the challenge of proving better bounds on the performance of product formulas, our strengthened bounds remain far from the apparent empirical performance. We expect that other ideas will be required to improve the product-formula approach \cite{CS19,LKW19}. Although our bounds have better asymptotic $n$-dependence than the previous commutator bound, they only offer an improvement if the system is sufficiently large. It could be fruitful to establish bounds for randomized product formulas that take advantage of the structure of the Hamiltonian, perhaps offering better performance both asymptotically and for small system sizes. More generally, it may be of interest to investigate other scenarios in which random choices can be used to improve the analysis of quantum simulation \cite{Campbell18,BCSWW19} and other quantum algorithms.

\section*{Acknowledgments}
We thank Guoming Wang for helpful discussions during the initial stages of this work and anonymous referees for their helpful comments on our manuscript.

This work was supported in part by the Army Research Office (MURI award W911NF-16-1-0349), the Canadian Institute for Advanced Research, the Department of Energy (grant 17-020469), and the National Science Foundation (grant 1526380).

\bibliographystyle{plainnat}
\bibliography{randsim}

\end{document}